\documentclass[conference]{IEEEtran}
\usepackage{geometry}
\usepackage[cmex10]{amsmath}
\usepackage{amssymb}
\usepackage{amsthm}
\usepackage{amsmath}
\usepackage{cite}
\usepackage{booktabs}
\usepackage{subfigure}
\usepackage{multicol}
\usepackage{graphicx}
\usepackage{epsfig}
\usepackage{epstopdf}
\usepackage{cases}
\usepackage{algorithmic}
\usepackage[ruled,linesnumbered]{algorithm2e}
\usepackage{bm}
\usepackage{multirow}
\usepackage{float}

\theoremstyle{definition} 
\theoremstyle{definition} 
\theoremstyle{plain} 
\theoremstyle{plain} \newtheorem{proposition}{Proposition}

\usepackage{mathrsfs}
\usepackage{threeparttable}

\usepackage{ifpdf}
\ifCLASSINFOpdf
\else
\fi
\begin{document}
%
\title{Joint Power Allocation and Caching Optimization in Fiber-Wireless Access Networks}

 \author{\IEEEauthorblockN{Zhuojia Gu, Hancheng Lu, Daren Zhu, Yujiao Lu }
 \IEEEauthorblockA{~Department of Electrical Engineering and Information Science \\
 University of Science and Technology of China, Hefei, Anhui 230027 China \\
  Email: guzj@mail.ustc.edu.cn, hclu@ustc.edu.cn, darenzhu@mail.ustc.edu.cn, lyj66@mail.ustc.edu.cn}}

\maketitle

%


\begin{abstract}
Fiber-Wireless (FiWi) access networks have been widely deployed due to the complementary advantages of high-capacity fiber backhaul and ubiquitous wireless front end. To meet the increasing demands for bandwidth-hungry applications, access points (APs) are densely deployed and new wireless network standards have been published for higher data rates. Hence, fiber backhaul in FiWi access networks is still facing the incoming bandwidth capacity crunch. In this paper, we involve caches in FiWi access networks to cope with fiber backhaul bottleneck and enhance the network throughput. On the other hand, power consumption is an important issue in wireless access networks. As both power budget in wireless access networks and bandwidth of fiber backhaul are constrained, it is challenging to properly leverage power for caching and that for wireless transmission to achieve optimal system performance. To address this challenge, we formulate the downlink wireless access throughput maximization problem by joint consideration of power allocation and caching strategy in FiWi access networks. To solve the problem, firstly, we propose a volume adjustable backhaul-constrained water-filling method (VABWF) to derive the expression of optimal wireless transmission power allocation. Then, we reformulate the problem as a multiple-choice knapsack problem (MCKP) and propose a dynamic programming algorithm to find the optimal solution of the MCKP problem. Simulation results show that the proposed algorithm significantly outperforms existing algorithms in terms of system throughput under different FiWi access network scenarios.
\end{abstract}

\hspace{8mm}
\par

%
\IEEEpeerreviewmaketitle

\section{Introduction}
In recent years, the growing bandwidth demand is foreseen to put tremendous pressure on access networks.
Cisco predicts that global Internet traffic will increase nearly threefold over the next 5 years, and video applications will account for 82\% of all Internet traffic by 2021, up from 73\% in 2016\cite{ref1}. Obviously, it is crucial for future access networks to support the growing demand for Internet traffic. To cope with the severe bandwidth crunch in access networks, fiber-wireless (FiWi) access networks were proposed to provide an efficient ``last mile'' Internet access in 1990s and have been widely deployed, especially in the field of multimedia communications \cite{ref2}. As an integration of optical access networks and wireless access networks, FiWi access networks leverage the complementary advantages of these two technologies and have attracted a great deal of research interest over the past two decades. Specifically, a FiWi access network combines optical fiber networks as its backhaul for high capacity and wireless networks as its front end for high flexibility and ubiquity \cite{high capacity and flexibility}.

In FiWi access networks, fiber backhaul bandwidth is shared by all wireless front ends. To meet the demands for bandwidth-hungry multimedia applications such as video on demand (VOD) services, access points (APs) are supposed to be densely deployed to improve network capacity in wireless front ends \cite{AP dense deployment}. Moreover, the recent IEEE 802.11ac standard specifies a data rate of at least 500Mbps for higher wireless access throughput. Hence, fiber backhaul in FiWi access networks is facing bandwidth bottleneck that limits the network throughput. Some researches have focused on next-generation passive optical networks (NG-PONs), which aim to achieve higher bandwidth in optical networks. But the cost of deploying NG-PONs is very high due to the fact that existing optical network units (ONUs) and optical line terminals (OLTs) are not directly compatible with the technical requirements of NG-PONs \cite{ref7}, so ONUs and OLTs must be replaced or upgraded to support NG-PONs, which is a costly method at present.

At the same time, there have been growing recent researches towards data caching as an approach for alleviating bandwidth pressure in fiber networks. In \cite{ref4}, an architecture consisting of an ONU with an associated storage unit was proposed to save traffic in the feeder fiber and improve the system throughput as well as mean packet delays. In \cite{ref12}, a framework of software-defined PONs with caches was introduced to achieve a substantial increase in served video requests. In \cite{ref13}, a dynamic bandwidth algorithm based on local storage VOD delivery in PONs was proposed and achievable throughput levels have been improved when a local storage is used to assist VOD delivery. However, most of them mainly focus on providing storage capacity in PONs to improve the network throughput, so these methods can not be directly applied to FiWi access networks.

In FiWi access networks, the functions of ONU and AP are integrated into a component called ONU-AP. To mitigate the bandwidth crunch of fiber backhaul in FiWi access networks, we equip ONU-APs with caches. Note that in FiWi access networks, the wireless front end will also make an impact on the system throughput in addition to the fiber backhaul network. Transmission power and channel conditions in the wireless front end should be taken into account, which is different from optical access networks with caches mentioned above. Particularly, transmission power of ONU-APs and channel conditions from user equipments (UEs) to ONU-APs are important factors determining the throughput in wireless access networks. When involving caches in FiWi access networks, it is challenging to properly leverage power for caching and that for wireless transmission to achieve optimal system performance. That is because if more power is used for wireless transmission, UEs can achieve higher wireless access throughput, whereas less power is used for caching, putting more pressure on fiber backhaul load. As both power budget in wireless access networks and bandwidth of fiber backhaul are constrained, joint power allocation and caching optimization are required to maximize the throughput in FiWi access networks.

In this paper, considering fiber backhaul bottleneck and channel conditions in wireless front end, as well as the sum power constraint of integrated ONU-APs, we formulate the joint power allocation and caching optimization problem as a Mixed Integer Programming (MIP) problem. To solve the problem, firstly, we propose a volume adjustable backhaul constrained water-filling method (VABWF) using convex optimization to derive the expression of optimal wireless transmission power allocation. Then, based on the derived expression, we reformulate the problem as a multiple-choice knapsack problem (MCKP) by exploiting the highest-popularity-first property of files. Finally, we propose a dynamic programming algorithm to solve the MCKP problem.

The rest of this paper is organized as follows. We first introduce the system model in Section \uppercase\expandafter{\romannumeral2}. Then, in Section \uppercase\expandafter{\romannumeral3}, the joint power allocation and caching optimization problem is formulated, and the transformation of the optimization problem into MCKP based on the proposed VABWF method is described. Simulation results are shown in Section \uppercase\expandafter{\romannumeral4} and concluding remarks are provided in Section \uppercase\expandafter{\romannumeral5}.

\section{System Model}
\begin{figure}
\includegraphics[width=3.4in]{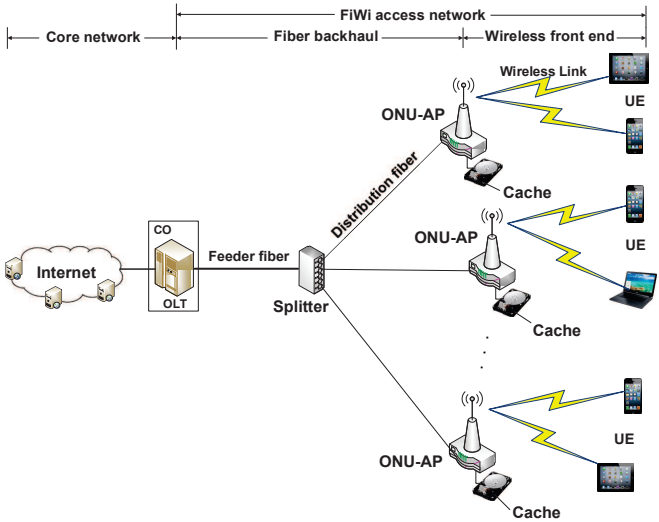}
\caption{FiWi access network architecture with caches at ONU-APs}
\label{fig:FiWi network architecture with caches at ONU-APs}
\end{figure}

\subsection{Network Model}
The network model we propose is shown in Fig.\ref{fig:FiWi network architecture with caches at ONU-APs}. The FiWi access network is divided into a fiber backhaul network and a wireless front end network. We adopt Passive Optical Network (PON) as the fiber backhaul network. The optical line terminal (OLT) is located at the central office (CO), connecting to an optical splitter through the feeder fiber. The optical splitter is connected to multiple ONUs through the distribution fiber. The split ratio of the optical splitter is usually 1:32 or 1:64. In the FiWi access network, each ONU is collocated with an AP, and the integration of the ONU and AP is called integrated ONU-AP. The ONU-APs are denoted by an index set $\mathcal{N}=\{1, 2, \cdots, n, \cdots,N\}$.

We assume that the feeder fiber connecting the splitter and the OLT has a capacity constraint of $C$.
At the wireless front end of the FiWi access network, $K$ user equipments (UEs) are randomly distributed in the coverage of ONU-APs. Let $\mathcal{K}$ denotes the UE index set and $\mathcal{K}=\{1, 2,\cdots, k,\cdots,K\}$. Each UE is assumed to be fixed associated with an ONU-AP, thus we denote the set of UEs associated with ONU-AP$_n$ by $\Phi_n,n \in \mathcal{N}$. We assume that ONU-APs adopt orthogonal frequency-division multiple access (OFDMA) to transmit data, that is to say, an ONU-AP allocates orthogonal subcarriers to UEs, and different transmission power levels are assigned to different UEs.

The received signal-to-interference-plus-noise ratio (SINR) at UE$_k$ associated with ONU-AP$_n$ can be expressed as

\begin{equation}\label{SINR}
  \gamma_{nk} = \frac{g_{nk}P_{nk}}{\sum\nolimits_{m\in\mathcal{N}\backslash\{n\}} g_{mk}P_{mk}+\sigma^2},
\end{equation}
where $N\backslash\{n\}$ denotes all ONU-APs except ONU-AP$_n$, $P_{nk}$ denotes the transmission power allocated to UE$_k$ by ONU-AP$_n$, $g_{nk}$ denotes the Rayleigh channel gain from ONU-AP$_n$ to UE$_k$ which follows the exponential distribution. $\sigma^2$ characterizes the background noise power level.

We assume that non-overlapping channels are assigned to adjacent ONU-APs, and ONU-APs that reuse the same spectrum are deployed far away from each other. Thus, the co-channel interference is negligible, namely, $g_{mk}=0, m\in \mathcal{N}\backslash \{n\}$. The wireless transmission rate from ONU-AP$_n$ to UE$_k$ is given by the Shannon's theorem
\begin{equation}\label{shannon}
  R_{nk}=B\log_{2}(1+\frac{g_{nk}P_{nk}}{\sigma^2}),
\end{equation}
where $B$ is the subchannel bandwidth for each UE.

\subsection{ONU-AP Caching Model}
Each ONU-AP is cache-enabled, and the cache size of ONU-AP$_n$ is denoted by $Q_n$. Requested files are denoted by an index set $\mathcal{J}=\{1,2,\cdots,j,\cdots,J\}$. If ONU-AP$_n$ already caches a requested file, it will respond to the request directly. Otherwise, the file should be fetched from the Internet via capacity-limited fiber backhaul. We assume all files have the same size, which is denoted by $s$. The file popularity distribution of $J$ files is modeled as Zipf distribution, so the probability of the $j$-th ranked file being requested by UEs is $p_{j}=\frac{j^{-\delta}}{\sum\nolimits_{j=1}^{J}j^{-\delta}}$,
where $\delta$ is a shape parameter that shapes the skewness of the popularity distribution. Let $p_n^{hit}$ denote the probability that files requested by UEs are cached at ONU-AP$_n$(i.e., cache hit ratio), and  $p_n^{miss}$ denote the probability that files requested by UEs are not cached at ONU-AP$_n$(i.e., cache miss ratio), where $p_n^{hit}+p_n^{miss}=1$. Given the cache decision on the $j$-th file at ONU-AP$_n$ as $x_{nj},x \in \{0,1\}$, the cache hit ratio at ONU-AP$_n$ can be written as
\begin{equation}\label{pnhit}
  p_{n}^{hit}=\sum\limits_{j \in \mathcal{J}}x_{nj}p_{j} =\frac{\sum\nolimits_{j \in \mathcal{J}}x_{nj}j^{-\delta}}{\sum\nolimits_{j \in \mathcal{J}}j^{-\delta}}.
\end{equation}

Note that files requested by UEs that are not cached at ONU-APs should be fetched from the Internet via fiber backhaul. According to the cache miss ratio  $p_n^{miss}$ at ONU-AP$_n$, the average fiber backhaul bandwidth occupied by UE$_k$ at ONU-AP$_n$ can be written as $p_n^{miss}R_{nk}$, thus the backhaul constraint of all $K$ UEs should be satisfied as
\begin{equation}\label{Rca-miss}
  \sum\limits_{n \in \mathcal{N}}\sum\limits_{k \in \Phi_n}p_n^{miss}R_{nk}\leq C.
\end{equation}

\subsection{Power Consumption Model}
The total power consumed by an ONU-AP$_n$ can be expressed by extending the typical wireless network power consumption model\cite{how much energy} to include caching power consumption as follows
\begin{equation}\label{Pntotal}
  P_n^{total}=\sum\limits_{k \in \Phi{n}}\rho P_{nk}+P_n^{ca}+P_n^{cc},
\end{equation}
where $P_n^{ca}$ denotes the power consumed at ONU-AP$_n$ for caching. $P_n^{cc}$ denotes the power consumed by circuits, which is a constant that depends on the circuit design. $\rho$ is a coefficient that measures the impact of power amplifier, power supply and cooling.
We use an energy-proportional model for caching power consumption \cite{power proportional model1} \cite{power proportional model2}, which has been widely adopted in content-centric networking for efficient use of caching power. In this model, the consumption of caching power is proportional to the total number of bits cached at an ONU-AP, which can be expressed as $P_n^{ca}=\omega \Omega_n$. $\Omega_n$ denotes the total number of bits cached at ONU-AP$_n$, and $\omega$ is a power coefficient related to caching hardware that reflects the caching power efficiency in watt/bit. In this paper, we consider the common caching device, high-speed solid state disk (SSD), for caching files at ONU-APs. The value of $\omega$ for SSD is $6.25 \times 10^{-12}$ watt/bit.

For each ONU-AP, it can adjust the transmission power and the caching power to achieve higher throughput without violating the maximum power constraint $P_M$, namely, $P_n^{total}\leq P_M$. For notational convenience, the circuit power $P_n^{cc}$ is omitted for it is a constant.


\section{Problem Formulation and Solution}
\subsection{Problem Formulation}
Our objective is maximizing the downlink wireless access throughput for UEs by optimizing power allocation and caching strategy at ONU-APs in the FiWi access network. This is formulated as follows,
\begin{subequations}
\begin{align}
\textbf{P1: } \max_{\bm{x},\bm{P}} \label{objective function}   &\quad\sum_{n\in \mathcal{N}}\sum_{k \in \Phi_{n}}B\log_2\left(1+\frac{g_{nk}P_{nk}}{\sigma^2}\right) \\
\text{s.t.}    &\quad \sum_{k \in \Phi_{n}}\rho P_{nk}+\omega\sum_{j \in \mathcal{J}}x_{nj}s\leq P_M\label{constraint:power_constraint},\forall n\in \mathcal{N}\\
        &\label{constraint:backhaul constraint}\quad \sum_{n \in \mathcal{N}}\sum_{k \in \Phi_{n}}\sum_{j \in \mathcal{J}}p_j(1-x_{nj})R_{nk} \leq C \\
        &\quad s\sum_{j \in \mathcal{J}}x_{nj} \leq Q_n, \quad \forall n\in \mathcal{N}\label{constraint:caching capacity constraint}\\
        &\quad P_{nk} \geq 0 , \quad \forall n\in \mathcal{N},\forall k \in \mathcal{K}\label{constraint:non-negative constraint}\\
        &\quad x_{nj}\in\{0,1\}, \quad  \forall n\in \mathcal{N},\forall j \in \mathcal{J} \label{constraint:binary constraint}
\end{align}
\end{subequations}

Constraint (\ref{constraint:power_constraint}) makes sure that the sum power consumed by each ONU-AP does not exceed the maximum power constraint $P_M$. Constraint (\ref{constraint:backhaul constraint}) ensures that the backhaul occupancy by uncached files should not be greater than the capacity constraint of fiber backhaul $C$. Constraint (\ref{constraint:caching capacity constraint}) makes sure that the total size of files cached at ONU-AP$_n$ should not exceed the cache capacity $Q_n$. Constraint (\ref{constraint:non-negative constraint}) guarantees that the transmission power is non-negative. Constraint (\ref{constraint:binary constraint}) means that the cache decision on the $j$-th file at ONU-AP$_n$ is a binary variable $x_{nj}$.

\subsection{Optimal Transmission Power Allocation}
Problem \textbf{P1} is a typical Mixed Integer Programming (MIP) problem, which is non-linear and non-convex. However, the problem becomes a convex optimization problem for fixed $\{x_{nj}^\circ\}$. By deriving the Karush-Kuhn-Tucker (KKT) conditions, we can get the optimal solution to transmission power allocation.

The Lagrangian of \textbf{P1} is given as
\begin{small}
\begin{align}\label{Lagrangian}
   &\quad \mathcal{L}(\bm{P},\lambda,\bm{\mu},\bm{\varepsilon})= \nonumber \\
   & -\sum_{n \in \mathcal{N}}\sum_{k \in \Phi_{n}}R_{nk}+ \varepsilon_{nk}P_{nk}
    +\lambda\Big(\sum_{n \in \mathcal{N}}\sum_{k \in \Phi_{n}}p_n^{miss}R_{nk}-C\Big)  \\
   & +\sum_{n \in \mathcal{N}}\mu_n\Big(\sum_{k \in \Phi_n}\rho P_{nk}+\omega\sum_{j \in \mathcal{J}}x_{nj} s-P_M\Big)
    -\sum_{n \in \mathcal{N}}\sum_{k \in \Phi_{n}}\varepsilon_{nk}P_{nk} \nonumber
\end{align}
\end{small}
where $\lambda,\bm{\mu} \in \mathbf{R}^N,\bm{\varepsilon} \in \mathbf{R}^{N\times K}$are Lagrangian multipliers.

\begin{proposition}\label{prop1}
Problem \textbf{P1} becomes a convex optimization problem for a given solution to the caching strategy $\{x_{nj}\}$.
\end{proposition}

\begin{proof}
The objective function (\ref{objective function}) is a strictly concave and increasing function with respect to $P_{nk}$, the inequality constraints (\ref{constraint:power_constraint}) and (\ref{constraint:backhaul constraint}) are convex for a given caching solution $\{x_{nj}^\circ\}$. Therefore, problem $\textbf{P1}$ becomes a convex optimization problem for fixed $\{x_{nj}^\circ \}$.
\end{proof}

The KKT conditions can be expressed as

\begin{subequations}
\begin{numcases}{}
\frac{\partial \mathcal{L}}{\partial P_{nk}}=\frac{g_{nk}}{\sigma^{2}+g_{nk}P_{nk}\ln2}(-1+\lambda \sum_{j \in \mathcal{J}}Bp_j(1-x_{nj}^\circ)) \nonumber\\
\qquad\qquad +\rho\mu_n-\varepsilon_{nk}=0 \label{necessary condition} \\
\lambda\bigg(\sum_{n \in \mathcal{N}}\sum_{k \in \Phi_{n}}\sum_{j \in \mathcal{J}}p_j(1-x_{nj}^\circ)R_{nk}-C\bigg)=0 \label{complementary slackness1}\\
\mu_n\bigg(\sum_{k \in \Phi_n}\rho P_{nk}+\omega\sum_{j \in \mathcal{J}}x_{nj}^\circ s-P_M\bigg)=0 \label{complementary slackness2}\\
\varepsilon_{nk}P_{nk}=0,\quad \forall n \in \mathcal{N},\forall k \in \mathcal{K} \label{complementary slackness3} \\
\lambda,\mu_n,\varepsilon_{nk}\geq 0,\quad \forall n \in \mathcal{N},\forall k \in \mathcal{K} \label{dual feasibility}
\end{numcases}
\end{subequations}
where (\ref{necessary condition}) is a necessary condition for an optimal solution, (\ref{complementary slackness1}), (\ref{complementary slackness2}) and (\ref{complementary slackness3}) represent the complementary slackness, and (\ref{dual feasibility}) represents the dual feasibility.

\begin{proposition}\label{prop2}
The downlink throughput in the FiWi access network is maximized when the ONU-APs consume the maximum power, i.e.,satisfy
\begin{equation}
 \sum_{k \in \Phi_{n}}\rho P_{nk}+\omega\sum_{j \in \mathcal{J}}x_{nj}s= P_M
 ,\quad \forall n \in \mathcal{N}\label{equality of power}
\end{equation}
\end{proposition}
\begin{proof}
Suppose that there exists $n \in N$ which satisfies $\sum_{k \in \Phi_{n}}\rho P_{nk}+\omega\sum_{j \in \mathcal{J}}x_{nj}s< P_M$, then ONU-AP$_n$ can increase its caching power to $P_M-\sum_{k \in \Phi_{n}}\rho P_{nk}$ for a higher cache hit ratio without reducing the sum rate of UEs associated to ONU-AP$_n$. Through proof by contradiction, we get the conclusion that the downlink throughput of the FiWi access network is maximized when the ONU-APs consume the maximum level of power.
\end{proof}

By using Proposition 1, 2 and the KKT conditions mentioned above, we can obtain an expression for the optimal transmission power $P_{nk}^*$ with respect to Lagrangian multipliers $\lambda,\mu_n$
\begin{equation}\label{optimal Pnk with lamda}
 P_{nk}^*=\bigg(\frac{B-\lambda B\sum_{j \in \mathcal{J}}p_j(1-x_{nj}^\circ)}{\rho\mu_n \ln2}-\frac{\sigma^2}{g_{nk}}\bigg)^+,
\end{equation}
with $(\cdot)^+=\max(\cdot,0)$.

\begin{proposition}
For problem \textbf{P1}, the optimal transmission power $P_{nk}^*$ allocated to UE$_k$ by ONU-AP$_n$ can be expressed as
\begin{equation}\label{optimal Pnk}
 P_{nk}^*=\bigg(\frac{B}{\rho\mu_n \ln2}-\frac{\sigma^2}{g_{nk}}\bigg)^+,
\end{equation}
and the corresponding Lagrangian multiplier $\mu_n$ can be expressed as

\begin{small}
\begin{equation}\label{mu n without lamda}
 \mu_n=\frac{|\Phi_n|B}{\left(P_M+\sum_{k \in \Phi_n}\frac{\sigma^2}{g_{nk}}-\omega\sum_{j \in \mathcal{J}}x_{nj}^\circ s\right)\ln2}.
\end{equation}
\end{small}

\end{proposition}

\begin{proof}
Eq.(\ref{optimal Pnk}) holds if the backhaul constraint (\ref{constraint:backhaul constraint}) is neglected (i.e., $\lambda=0$). Substituting Eq.(\ref{optimal Pnk with lamda}) into Eq.(\ref{equality of power}), we obtain Eq.(\ref{mu n without lamda}). What we need to prove is that the conclusion still holds with the backhaul constraint (\ref{constraint:backhaul constraint}). Suppose the optimal solution to problem \textbf{P1} is $\{P_{nk}^*,x_{nj}^*\}$, then $\sum_{k \in \Phi_n}\rho P_{nk}^* +\omega\sum_{j \in \mathcal{J}}x_{nj}^* s=P_M$. Let $\{P_{nk}^\circ,x_{nj}^\circ\}$ denote the solution that satisfies Eq.(\ref{optimal Pnk}), which indicates the maximum sum rate of UEs associated to ONU-AP$_n$ without the backhaul constraint and is expressed as

\begin{small}
\begin{subequations}
\begin{numcases}{}
\sum_{k \in \Phi_n}\log_2\Big(1+\frac{g_{nk}P_{nk}^\circ}{\sigma^2}\Big)\geq\sum_{k \in \Phi_n}\log_2\Big(1+\frac{g_{nk}P_{nk}^*}{\sigma^2}\Big) \label{contradiction1}\qquad \\
\sum_{k \in \Phi_n}P_{nk}^\circ \leq \sum_{k \in \Phi_n}P_{nk}^*  \label{contradiction2}
\end{numcases}
\end{subequations}
\end{small}

From Eq.(\ref{equality of power}) and Eq.(\ref{contradiction2}) we obtain $p_n^{hit^\circ} \geq p_n^{hit^*}$, which means that $\{P_{nk}^\circ,x_{nj}^\circ \}$ gets a higher cache hit ratio without reducing the sum rate of UEs associated to ONU-AP$_n$. Therefore, Eq.(\ref{contradiction1}) also holds with the backhaul constraint, so $\{P_{nk}^\circ,x_{nj}^\circ \}$ yields the optimal solution which is obtained by using Eq.(\ref{optimal Pnk}).
\end{proof}


In Eq.(\ref{optimal Pnk}), $\frac{\sigma^2}{g_{nk}}$ is the inverse of channel gain normalized by the noise variance $\sigma^2$. Eq.(\ref{optimal Pnk}) complies with the form of water-filling method \cite[Chapter 6]{waterfilling}.  Think of a vessel whose bottom is formed by plotting those values of $\frac{\sigma^2}{g_{nk}}$ for each subchannel $k$. Then, we flood the vessel with water to a depth $\frac{B}{\rho \mu_n \ln2}$. Note that the volume of water is not fixed and it depends on the caching solution $\{x_{nj}^\circ\}$. For a given caching solution, the total amount of water used is then $P_M-(\omega \sum_{j \in \mathcal{J}}x_{nj}^\circ s)$. The depth of the water at each subchannel $k$ is equal to the optimal transmission power allocated to the channel. Proposition 3 guarantees that the optimal solution to transmission power allocation still holds with the backhaul constraint, so this method is called volume adjustable backhaul-constrained water-filling method (VABWF).

\newtheorem*{Theorem 1}{Theorem 1}
\begin{Theorem 1}
On the condition of total power constraint, files with higher popularity should be cached preferentially for maximizing the throughput in the FiWi access network.\label{Theorem1}
\end{Theorem 1}

\begin{proof}
Suppose that $\{P_{nk}^*,x_{nj}^* \}$ is the optimal solution to problem \textbf{P1} with the corresponding downlink throughput $R^{dl^*}$, and $\{x_{nj}^* \}$ violates the popularity based caching policy mentioned in Theorem 1. Let $\{P_{nk}^\circ,x_{nj}^\circ \}$ be another set and $\{x_{nj}^\circ \}$ satisfies

\begin{small}
\begin{equation}\label{prop3 prove1}
  \Big\{x_{nj}^\circ \Big| x_{nj}^\circ \geq x_{n(j+1)}^\circ,\sum_{j \in \mathcal{J} } x_{nj}^\circ = \sum_{j \in \mathcal{J} } x_{nj}^* , \forall n \in \mathcal{N}, \forall j \in \mathcal{J} \Big\}. \nonumber
\end{equation}

\end{small}

Namely, $\{x_{nj}^\circ \}$ is a caching solutions that obey Theorem 1. As $p_j$ is monotonically decreasing in $j$, the cache miss ratio satisfies
\begin{small}

\begin{equation}\label{prop3 prove2}
  \sum_{j \in \mathcal{J}} p_j(1-x_{nj}^\circ) < \sum_{j \in \mathcal{J}} p_j(1-x_{nj}^*).
\end{equation}
\end{small}

Note that $\{P_{nk}^\circ,x_{nj}^\circ \}$ satisfies $\sum_{j \in \mathcal{J}} x_{nj}^\circ = \sum_{j \in \mathcal{J}} x_{nj}^*$. With constraint (\ref{prop2}), we can obtain a set of transmission power solution that satisfy $P_{nk}^\circ = P_{nk}^*$. With (\ref{prop3 prove2}), $\{P_{nk}^\circ,x_{nj}^\circ \}$ is ensured to satisfy the constraint (\ref{constraint:backhaul constraint}), so $\{P_{nk}^\circ,x_{nj}^\circ \}$ is a solution to problem \textbf{P1}, with the corresponding downlink throughput $R^{dl^\circ}$. By the optimality of $\{P_{nk}^*,x_{nj}^* \}$ to problem \textbf{P1}, we have
\begin{equation}\label{prop3 prove3}
  R^{dl^*} - R^{dl^\circ} \geq 0.
\end{equation}
Since $P_{nk}^\circ = P_{nk}^*$, the equality in (\ref{prop3 prove3}) must hold, which concludes that $\{P_{nk}^\circ,x_{nj}^\circ \}$ yields the optimal solution to problem \textbf{P1}.
\end{proof}

\subsection{Problem Reformulation and Solution}
By using Theorem 1 and the VABWF method mentioned previously, the solution space of problem \textbf{P1} is greatly reduced. We denote the solution space of problem \textbf{P1} by $\mathcal{A}$, thus we have that
\begin{align}\label{definition of set A}
   \mathcal{A} =  \bigg\{ \left\{P_{nk},x_{nj}\right\} \bigg | & P_{nk}=\bigg(\frac{B}{\rho\mu_n \ln2}-\frac{\sigma^2}{g_{nk}}\bigg)^+; \nonumber \\
  &  x_{nj} \geq x_{n(j+1)},j \in \mathcal{J} \bigg\} \nonumber
\end{align}
The element of $\mathcal{A}$, denoted $A_{nj} = \{P_{nk}^\circ,x_{nj}^\circ \} \in \mathcal{A}$, should satisfy
\begin{equation}\label{set A satisfication}
  \sum_{j \in \mathcal{J}}x_{nj}^\circ = j.
\end{equation}
Let $\mu(A_{nj})$ denote whether $A_{nj}$ is chosen to be the solution, which can be expressed as
\begin{numcases}{\mu(A_{nj})=}
1, A_{nj} \text{ is chosen to be the solution} \nonumber  \\
0, \text{ otherwise} \nonumber
\label{mu Anj definition}
\end{numcases}
Define $\omega(A_{nj})$ to be the backhaul bandwidth occupied by ONU-AP$_n$ with respect to solution $A_{nj}$
\begin{numcases}{\omega(A_{nj})=}
\sum_{k \in \Phi_n}p_n^{miss}R_{nk}\big|_{\{P_{nk},x_{nj}\}=A_{nj}}, \mu(A_{nj})=1 \nonumber \\
\ 0 \  \qquad \qquad \qquad \qquad \qquad \quad ,  \mu(A_{nj})=0
\label{omega Anj definition} \nonumber
\end{numcases}
Likewise, define $\nu(A_{nj})$ to be the sum rate of UEs associated to ONU-AP$_n$ with respect to solution $A_{nj}$
\begin{numcases}{\nu(A_{nj})=}
\sum_{k \in \Phi_n}R_{nk}\big|_{\{P_{nk},x_{nj}\}=A_{nj}}, \mu(A_{nj})=1 \nonumber \\
\ 0\  \qquad \qquad \qquad \qquad \ \ ,  \mu(A_{nj})=0
\label{nu Anj definition} \nonumber
\end{numcases}

Then problem \textbf{P1} can be converted into the problem of determining the value of $\mu(A_{nj})$ with the aim to maximize the downlink wireless access throughput of the FiWi access network as follows,
\begin{subequations}
\begin{align}
\textbf{P2: } \max_{\mu(A_{nj})} \label{objective function2}   &\quad\sum_{n\in \mathcal{N}}\sum_{j\in \mathcal{J}}\nu(A_{nj}) \\
\text{s.t.}    &\label{constraint:backhaul constraint2}\quad \sum_{n \in \mathcal{N}}\sum_{j \in \mathcal{J}}\omega(A_{nj}) \leq C \\
        &\label{constraint:multiple choice constraint}\quad \sum_{j \in \mathcal{J}} \mu(A_{nj}) \leq 1 ,\quad \forall n\in \mathcal{N},\forall j \in \mathcal{J} \\
        & \ \quad \mu(A_{nj}) \in \{0,1\}, \quad \forall n\in \mathcal{N},\forall j \in \mathcal{J}
\end{align}
\end{subequations}

Problem \textbf{P2} is in the form of a multiple-choice knapsack problem (MCKP)\cite[Chapter 11]{MCKP}. The problem is to choose no more than one item from each class such that the profit sum is maximized without exceeding the capacity $C$ in the corresponding backhaul limitation. Considering an optimal solution to the MCKP problem, it is obvious that by removing any class $n$ from the optimal MCKP packing, the remaining solution set must be an optimal solution to the subproblem defined by capacity $C-\omega(A_{nj})$ and class set $\mathcal{N}\backslash \{n\}$. Any other choice will risk to diminish the optimal solution value. Hence, problem \textbf{P2} has the property of an optimal substructure as described in \cite[Section 15.3]{introduction to algorithms}. We design an optimal and efficient algorithm through dynamic programming based on VABWF method as outlined in Algorithm 1.
Define $R(n,c)$ to be the maximum downlink throughput of the FiWi access network on the condition that there exist only the first $n$ classes with backhaul limitation $c$. Then we can consider an additional class to calculate the corresponding maximum throughput and the following recursive formula describe how the iteratively method is performed
\begin{align}\label{recursive formula}
  R(n,c) = & \max \bigg\{ \big\{R(n-1,c-\omega(A_{nj}))+ \nu(A_{nj}) \big\} \cup \nonumber \\
   &  \big\{R(n-1,c)\big\}     \Big| c-\omega(A_{nj}) \geq 0, \forall j \in \mathcal{J}  \bigg\}.
\end{align}
Note that the constraint (\ref{constraint:multiple choice constraint}) is satisfied by placing the recursive formula in the innermost loop. At each iteration, we choose the optimum solution to the given number of classes $n$ and bandwidth limitation $c$. The running time of Algorithm 1 is dominated by the $c$ iterations of the second for-loop, each of which contains at most $J$ iterations where a new solution of a subproblem is computed. Considering $N$ ONU-APs in the FiWi access network, there are $N$ subproblems to be computed, so the overall time complexity is $O(NCJ)$.
\renewcommand{\algorithmicrequire}{\textbf{Input:}}
\renewcommand{\algorithmicensure}{\textbf{Output:}}

\begin{algorithm}
\caption{\mbox{Dynamic Programming Algorithm for \textbf{P2}}}
\label{alg:Optimal Algorithm}
\begin{algorithmic}[1]
\REQUIRE{$~\mathcal{N}$,$~\mathcal{K}$,$~\mathcal{J}$,$~B$,$~P_M$,$~C$,$~Q_n$,$~g_{nk}$,$~\omega$,$~s$,$~\rho$,$~\sigma$,$~\delta$;}

\ENSURE{$\{P_{nk}^\ast,x_{nj}^\ast\},$ and the downlink  throughput $R^{dl}$;}
\FORALL {$~A_{nj} \in \mathcal{A},\forall n\in \mathcal{N},\forall j\in \mathcal{J}$}
\STATE {Calculate $\mu_n$ according to (\ref{mu n without lamda});}
\STATE {$P_{nk} \Leftarrow \frac{B}{\rho\mu_n\ln2}-\frac{\sigma^2}{g_{nk}}$;}
\WHILE {$\exists P_{nk}<0,\forall k \in \Phi_n$}
\STATE {$P_{nk}\Leftarrow (P_{nk})^+$;}
\STATE {Recalculate $\mu_n$ and reallocate $P_{nk}$ to those $P_{nk}>0$;}
\ENDWHILE
\STATE {Calculate $\omega(A_{nj}), \nu(A_{nj})$;}
\ENDFOR
\STATE {Calculate $R(1,c), c=0,1,2,\cdots,C$;}
\FOR {$n=2;n\leq N;n++$}
\FORALL {$c\in\big \{0,1,2,\cdots,C \big\}$}
\IF {$c-min\big \{\omega(A_{nj})|\forall j \in \mathcal{J}\big\}<0$}
\STATE {$R(n,c)\Leftarrow R(n-1,c)$;}
\ELSE
\STATE {Update $R(n,c)$ according to (\ref{recursive formula});}
\ENDIF
\IF {$R(n,c)\neq R(n-1,c)$}
\STATE {$tr(n,c) \Leftarrow argmax_j\big\{\{R(n-1,c)\}\cup\newline\{R(n-1,c-\omega(A_{nj}))+\nu(A_{nj})\}\big\}$;}
\ENDIF
\ENDFOR
\ENDFOR
\STATE { Find the optimal solution $A_{nj}$ according to \\
 $tr(n,c)$ in the reverse order of n;}
\RETURN {$\{P_{nk}^\ast,x_{nj}^\ast\},$ $R^{dl}$;}
\end{algorithmic}
\end{algorithm}

\section{Simulation Results}
\begin{table}
  \centering
  \caption{SIMULATION SETTING}
  \label{simulationsetting}
  \begin{tabular}{|p{3.6cm}|p{4.2cm}|}
    \hline
    \textbf{Parameter}                              & \textbf{Value}\\\hline
    Number of ONU-APs                            &$32$ \\\hline
    UE numbers at each ONU-AP  &$20$\\\hline
    System bandwidth                   &$20$MHz\\\hline
    Subchannel bandwidth               &$500$kHz\\\hline
    Thermal noise density&$-174$dBm/Hz\\\hline
    Number of files& $1000$\\\hline
    Size of file& $100$MB\\\hline
    Cache size of each ONU-AP  &$30$GB\\\hline
    Caching power efficiency   &$6.25\times 10^{-12}$W/bit  \\\hline
    Circuit power at each ONU-AP   &$3$W   \\\hline
    Transmission power coefficient   & $1.2$\\\hline
    Fiber backhaul capacity  &$2.488$Gbps (unless stated otherwise)\\\hline
    Maximum total power for each ONU-AP & $7$W (unless stated otherwise)\\\hline
    Zipf parameter  & $0.8$ (unless stated otherwise)\\\hline

  \end{tabular}
\end{table}
In this section, we validate the performance of our proposed caching and power allocation algorithm under different FiWi access network scenarios. The simulation parameters are summarized in Table \uppercase\expandafter{\romannumeral1}. For a comparison purpose, we introduce full-cache, equal-power and random algorithms as follows,

\begin{figure*}[ht]
  \centering
  \subfigure[Throughput vs. fiber backhaul capacity ]{
    \label{fig:colt} 
    \includegraphics[width=2.3in]{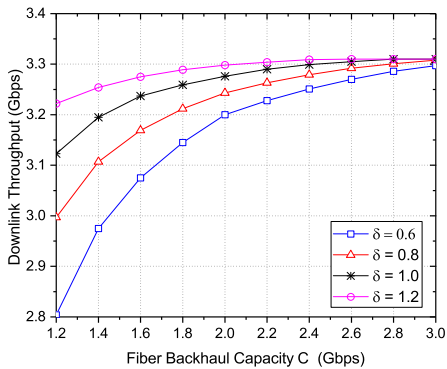}}
  \hspace{0.7in}
  \subfigure[Cache utilization vs. fiber backhaul capacity ]{
    \label{fig:utilization} 
    \includegraphics[width=2.3in]{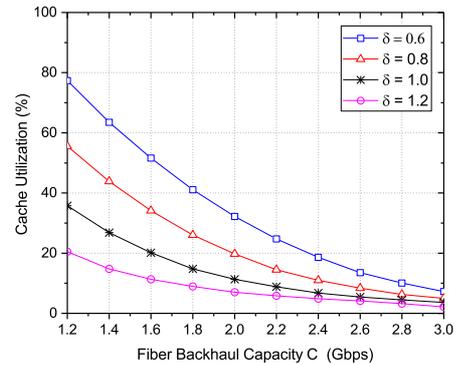}
    }
  \caption{Throughput and cache utilization of the proposed algorithm }
  \label{fig:performance of the proposed method} 
\end{figure*}

\begin{itemize}
  \item \emph{Full-Cache}: Each ONU-AP fully uses the cache capacity to store the most popular files, and the classical water-filling method is used to allocate transmission power to UEs.
  \item \emph{Equal-Power}: Transmission power is equally allocated to UEs. Meanwhile, each ONU-AP choose the most popular files to cache, but does not necessarily use full cache capacity.
  \item \emph{Random}: ONU-APs randomly choose files to cache and transmission power is also equally allocated to UEs.
\end{itemize}

\subsection{Downlink Throughput and Cache Utilization of the Proposed Algorithm}
As shown in Fig.\ref{fig:colt}, observing typical GPON downlink rate of 2.488Gbps, the corresponding downlink throughput is increased by 30\% with regard to the backhaul bandwidth. Even greater improvement of the downlink throughput is observed when the backhaul bandwidth is 1.25Gbps, which is a typical downlink rate in EPON. The throughput is increased by about 2.5 times with regard to the backhaul bandwidth. As for the cache utilization shown in Fig.\ref{fig:utilization}, we never use full cache capacity and the cache utilization is below 50\% in most instances. Besides, the slope of the curves changes with different values of $\delta$.

Fig.\ref{fig:colt} and Fig.\ref{fig:utilization} show that the downlink throughput increases with the increase of fiber backhaul capacity, whereas the cache utilization reduces as the backhaul capacity increases. This is not surprising, because the larger the fiber backhaul capacity is, the more cache miss files can be fetched via the feeder fiber from the Internet. Thus, less power for caching is used for local storage. On the other hand, the cache utilization is higher when the Zipf distribution parameter $\delta$ has a smaller value, which means the proposed algorithm choose to cache more files rather than increase transmission power for the purpose of maximizing the throughput. This is due to the fact that the file popularity becomes more decentralized with smaller value of $\delta$, so the proposed algorithm tends to cache more files to avoid the backhaul bottleneck.

\subsection{Comparisons with the other Algorithms}
\begin{figure*}[ht]
  \centering
  \subfigure[Throughput vs. maximum total power]{
    \label{fig:Throughput vs. maximum total power} 
    \includegraphics[width=2.2in]{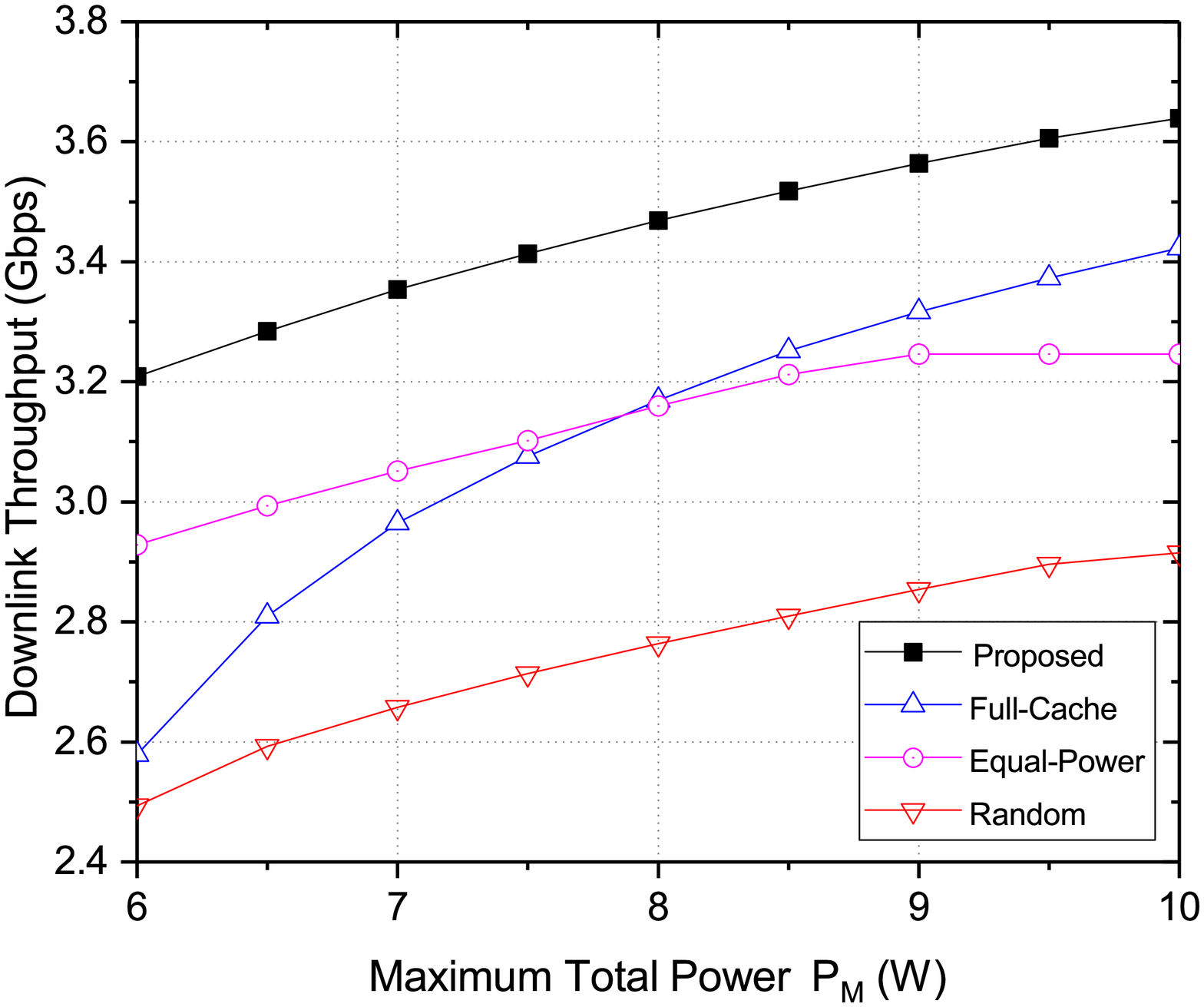}}
  \subfigure[Throughput vs. fiber backhaul capacity]{
    \label{fig:Throughput vs. fiber backhaul capacity} 
    \includegraphics[width=2.2in]{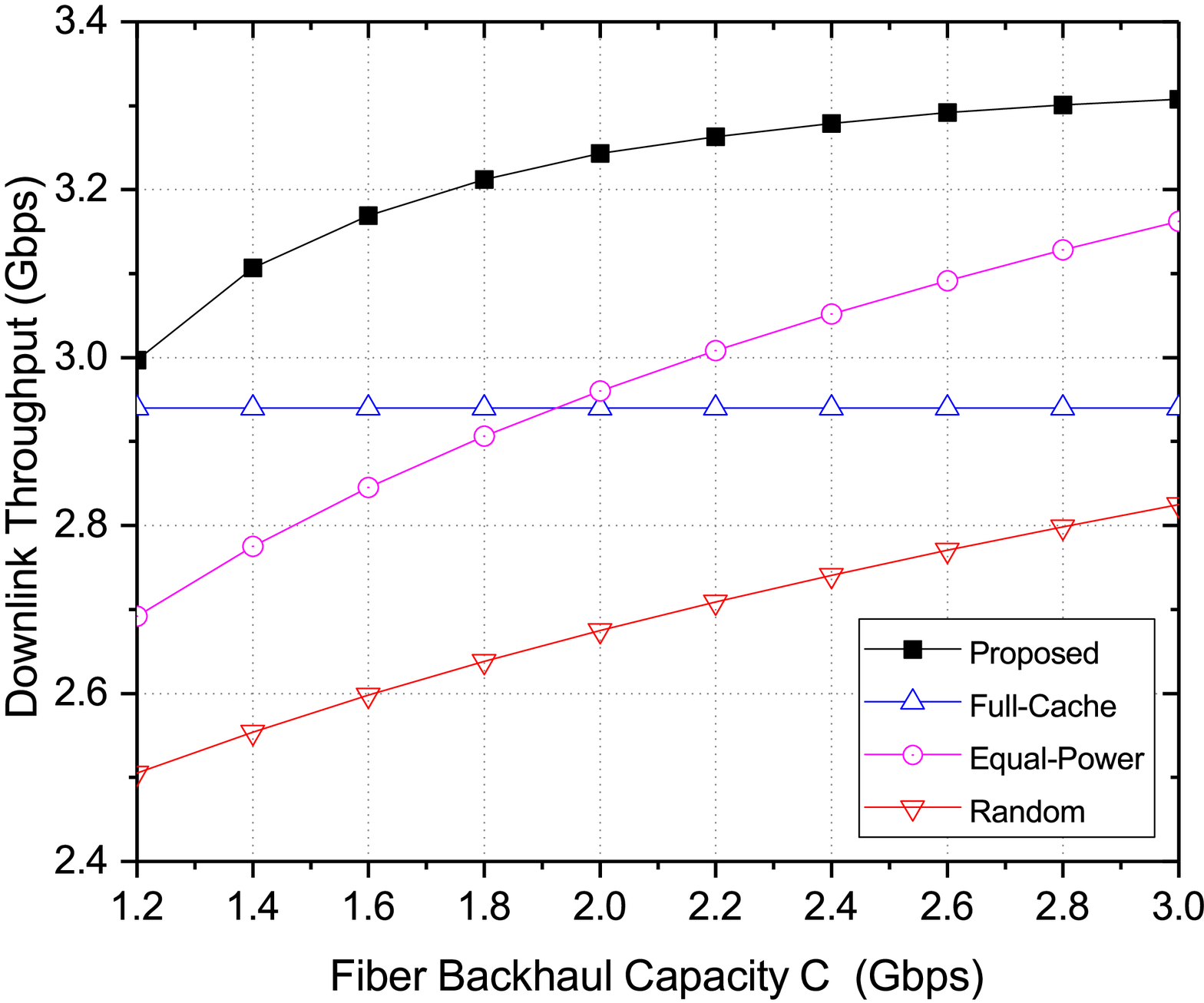}
    }
  \subfigure[Throughput vs. diversity of file popularity]{
    \label{fig:Throughput vs. diversity of file popularity} 
    \includegraphics[width=2.2in]{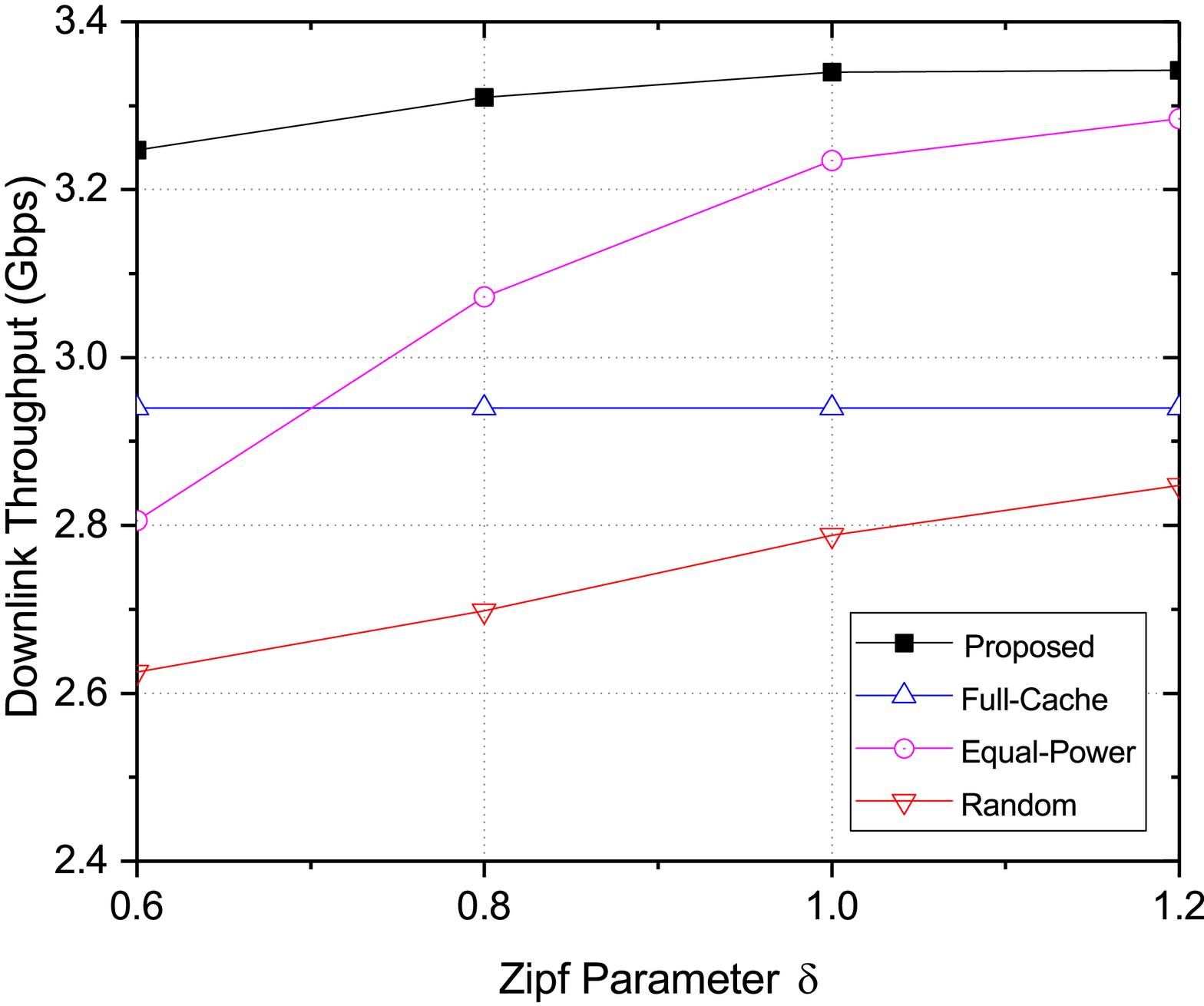}
    }
  \caption{Comparison between the proposed algorithm and the other algorithms}
  \label{fig:Numerical result_aggregated_window} 
\end{figure*}
Fig. \ref{fig:Throughput vs. maximum total power} illustrates the downlink throughput with respect to maximum total power $P_M$. It is observed that the proposed algorithm outperforms the others with an increase of 10.4\%, 11.8\% and 25.8\% in the downlink throughput, respectively. This is because the proposed algorithm can adaptively adjust the transmission power and caching power under different maximum total power values. Under the condition that the backhaul occupied by cache miss files does not exceed the backhaul limitation $C$, our proposed algorithm tries to increase the transmission power level and allocate more power to UEs with better channel conditions for higher throughput. We can observe that the interval between the proposed algorithm and full-cache algorithm becomes smaller with the increase of maximum total power available for ONU-APs. This is due to the fact that by adopting the proposed algorithm, more caching power is used when the value of maximum total power increases with the purpose of maximizing downlink wireless access throughput. Therefore, the solution to the problem using full-cache algorithm gets closer to the optimal solution using the proposed algorithm. In addition, equal-power algorithm fails to get higher throughput because the transmission power allocated to UEs is a fixed value, and the bottleneck appears when higher value of maximum total power is chosen. The worst performance occurs when random algorithm is adopted because the highest-popularity-first caching strategy is not used. Thus, the cache hit ratio is rather low, which limits the system throughput.

Fig. \ref{fig:Throughput vs. fiber backhaul capacity} illustrates the downlink throughput with respect to fiber backhaul capacity $C$. It is again observed that the proposed algorithm has much better performance than the others. It is worth noting that the downlink throughput using full-cache algorithm is considerable when the backhaul capacity $C$ is small, but never increases even if the backhaul capacity increases to a higher level. This is not surprising because the total transmission power is fixed when we use full-cache algorithm. Therefore, the throughput is never enhanced regardless of the increased backhaul capacity. On the contrary, our proposed algorithm makes full use of the backhaul capacity while increasing the transmission power level for higher throughput.

Fig.  \ref{fig:Throughput vs. diversity of file popularity} illustrates the downlink throughput with respect to diversity of file popularity. The proposed algorithm performs better than the others under different values of Zipf parameter $\delta$. We observe that the throughput increases with the increase of $\delta$. This is because more requests centralize on a few files with the rest of files rarely getting requested when $\delta$ is larger, thus our proposed algorithm tends to cache less files and save more caching power for increasing transmission power to achieve higher throughput. The interval between our proposed algorithm and equal-power algorithm gets smaller with the increase of parameter $\delta$, but is not expected to be eliminated. This is reasonable because the cache hit ratio gets higher when parameter $\delta$ increases, leading to better performance of equal-power algorithm. Nevertheless, with transmission power equally allocated to UEs, the optimal solution will never be obtained. We also notice that the throughput remains unchanged using full-cache algorithm despite the change of parameter $\delta$. This indicates a waste of caching power because files rarely being requested are cached in ONU-APs. Therefore, the transmission power drops to a low level and the downlink throughput is not expected to increase.

\section{Conclusion}
In this paper, we equip ONU-APs with caches in FiWi access networks to deal with fiber backhaul bottleneck and further enhance the system throughput. To achieve the optimal downlink wireless access throughput, we propose a dynamic programming algorithm based on VABWF to perform joint power allocation and caching optimization. Simulation results show that the proposed algorithm outperforms full-cache and random algorithms as well as equal-power allocation algorithm significantly in terms of system throughput.

\section*{ACKNOWLEDGMENT}
This work was supported in part by the National Science Foundation of China (No.61390513, 91538203, 61771445) and the Fundamental Research Funds for the Central Universities.

\end{document}